\documentclass[12pt]{article}
\usepackage{amsmath,amssymb, amsthm, amsfonts,graphics,epsfig,calc,color,array}
{\begin{Sbox}\begin{minipage}}%
{\end{minipage}\end{Sbox}\fbox{\TheSbox}}%

\usepackage{tabularx}
\usepackage{pifont}
\usepackage{breakurl}
\usepackage{cite}
\usepackage{natbib}
\usepackage[perpage,symbol]{footmisc}

\usepackage[colorlinks,citecolor=blue,urlcolor=blue]{hyperref}

\newtheorem{theorem}{Theorem}
\newtheorem{lemma}{Lemma}

\newtheorem{remark}{Remark}

\begin{document}

\title{
\bf  Ranking in the generalized Bradley-Terry models when the strong connection condition fails}
\author{Ting Yan \\
\small Department of Statistics, Central China Normal University, Wuhan, China
}
\date{}
\maketitle

\setcounter{footnote}{-1}
\footnote{Address correspondence to Ting Yan,
Department of Statistics,
Central China Normal University,
Wuhan 430079, China;
E-mail: tingyanty@mail.ccnu.edu.cn}

\begin{abstract}
\baselineskip=16pt
\noindent
For nonbalanced paired comparisons, a wide variety of ranking methods have been proposed.
One of the best popular methods is the Bradley-Terry model in which the ranking of
a set of objects is decided by the maximum likelihood estimates (MLEs) of merits parameters.
However, the existence of MLE for the Bradley-Terry model and its generalized models
to allow for tied observation or home-field advantage or both to occur,
crucially depends on the strong connection condition on the directed graph
constructed by a win-loss matrix.
When this condition fails, the MLE does not exist and hence there is no solution of ranking.
In this paper, we propose an improved version of the $\varepsilon$ singular perturbation
proposed by Conner and Grant (2000), to address this problem and extend it to the generalized Bradley-Terry models.
Some necessary and sufficient conditions for the existence and uniqueness
of the penalized MLEs for these generalized Bradley-Terry-$\varepsilon$ models are derived.
Numerical studies show that the ranking is robust to the different $\varepsilon$.
We apply the proposed methods to the data of the 2008 NFL regular season.

\end{abstract}

\vskip 20 pt

\section{Introduction}
For nonbalanced paired comparisons, a wide variety of ranking methods have been proposed
[see, e.g., B\"{u}hlmann and Huber (1963); David (1987); Keener (1993)].
Among them, the Bradley-Terry model that itself
was independently proposed by Zermelo (1929) and Ford (1957), in which the ranking is decided by
the maximum likelihood estimates (MLEs) of merits parameters,
is one of the best popular methods.
For example, the World Chess Federation and the European Go
Federation use it to rank players. The Elo system uses it to
estimate the winning probability between two players.

Assume that $t$ objects are engaged in paired comparisons experiments
in which one object is judged to be preferred to another object in each comparison.
Let $a_{ij}$ be the number of times that $i$ wins $j$. By notation, we define $a_{ii}=0$.
Bradley and Terry (1952) supposed that the probability of $i$ winning $j$ is
\begin{equation*}
p_{ij}=\frac{u_i}{u_i+u_j},~~~~i,j=1,\ldots,t; i\neq j,
\end{equation*}
where $u_i$ is the merit parameter of object $i$. The likelihood function is
\begin{equation}\label{original-likelihood}
\prod_{i,j=1}^t \left(  \frac{u_i}{u_i+u_j} \right)^{a_{ij}}.
\end{equation}
Since \eqref{original-likelihood} is scaled invariable, there are two equivalently
normalized ways: $\sum_{i=1}^t u_i=1$ or $u_1=1$.
In order to guarantee the existence and uniqueness of the MLE for
the likelihood \eqref{original-likelihood},
the following strong connection condition is necessary and sufficient due to Ford (1957).

Condition A. The directed graph $\mathcal{G}(A)$ constructed by the win-loss matrix $A=(a_{ij})$
is strongly connected.

There are two equivalent statements for Condition A. One is that
the win-loss matrix $A=(a_{ij})$ is irreducible in the algebra language.
The other is that for any partition of all objects into two nonempty sets, some object in second set has
beaten some object in the first set at least once.

The Bradley-Terry model and its extensions have been extensively discussed in the
literature [for wide surveys see Davidson and Farquhar (1976), David (1988), Simons and Yao (1999)].
Hunter (2004) established minorization-maximization  (MM) algorithms for a class of extended Bradley-Terry models
to solve the MLEs. Recently, Caron and Doucet (2012) derived
expectation-maximization (EM) algorithms for solving the maximum a posteriori estimate (MAPE) by incorporating some
suitable latent variables into the Bradley-Terry model and its extensions, and shew
that Hunter's MM algorithms can be reinterpreted as the standard EM algorithms,
and first proposed Gibbs samplers to perform Bayesian inference.
When the number of objects $t$ is small and all $n_{ij}$ are relatively large,
where $n_{ij}$ is the number of comparisons between $i$ and $j$ and $n_{ii}=0$ for convenience,
Condition A usually holds.
In its inverse scenario that $t$ is moderate or large and the design matrix $\mathbf{n}=(n_{ij})$
is sparse, Condition A may fail.
A simple example is that whenever there are undefeated objects or objects with no wins, Condition A fails.
The game results of 2007 and 2008 National Football League (NFL) regular seasons
in which there are $32$ teams and each team plays with only 13 other teams, fall into this case.
More cases can be found in the American college football regular seasons,
in which hundreds of teams play each other and each team has at most tens of games.
When Condition A fails, the MLE doesn't exist not only in the Bradley-Terry model but also in the
the Thurstone model suggesting the probability of $i$ preferred to $j$ is $\Phi(\log u_i - \log u_j)$
where $\Phi(\cdot)$ is the standard normal distribution (Thurston, 1927).
Several approaches have been proposed to address this problem.
Mease (2003) used the penalized likelihood with the penalized factor
$\prod_i \Phi( \log u_i )\Phi(-\log u_i)$ in the Thurstone model to rank American college football teams.
For the Bradley-Terry model, Conner and Grant (2000) proposed a singular perturbation
by letting $\bar{a}_{ij}=a_{ij}+\varepsilon(1-\delta_{ij}),~ \varepsilon>0$,
where $\delta_{ij}$ is the Kronecker delta sign, and proved that the extended ranking is identical to that of
the original model if Condition A holds.
This data transformation guarantees that the directed graph $\Gamma(\bar{A})$ is strongly connected
in any case and therefore the solution maximizing the function \eqref{original-likelihood}
with $a_{ij}$ replaced by $\bar{a}_{ij}$, uniquely exists.
Another possibility is to use the method of Bayesian ranking by Caron and Doucet (2012).

Different $\varepsilon$ may lead to different rankings.
For example, consider such outcomes $a_{12}=2$, $a_{21}=1$, $a_{14}=1$, $a_{34}=1$, $a_{43}=2$
and for others $a_{ij}=0$.
When $\varepsilon=0.1$, $\bar{u}_1=1$, $\bar{u}_2=0.470$, $\bar{u}_3=0.162$, $\bar{u}_4=0.262$,
where $(\bar{u}_1, \ldots, \bar{u}_t)$ is the solution maximizing the function \eqref{original-likelihood}
with $a_{ij}$ replaced by $\bar{a}_{ij}$.
When $\varepsilon=0.5$, $\bar{u}_1=1$, $\bar{u}_2=0.569$, $\bar{u}_3=0.453$, $\bar{u}_4=0.585$.
Therefore $\varepsilon=0.1$ and $\varepsilon=0.5$ lead two different rankings $1\succ2\succ4\succ3$ and $1\succ4\succ2\succ3$,
where $i \succ j$ denotes the partial order of $i$ and $j$ if $\bar{u}_i>\bar{u}_j$.
In this paper, we propose an improved singular perturbation method that aims to eliminate this possibility in Section 2,
and generalize it into
a class of extended Bradley-Terry models involving the ties and home-field advantage in Section 3.
Numerical studies and the application of the 2008 NFL regular season
are given in Section 4.
In section 5, the comparisons between the methods of Bayesian ranking
and the perturbation method are presented.
Some discussion is in Section 6.

\section{Improved singular perturbation method}
If all objects can be partitioned into two nonempty sets so
that there are no inter-set comparisons, then there is no basis for
ranking the objects in the first set against the objects in the second set.
The following condition may be a minimum requirement for any paired comparison design due to Kendall (1955).

Condition B. The undirected graph constructed by the adjacent matrix $\mathbf{n}$ is connected.

Condition B is equivalent to that $\mathcal{G}(A)$ is weakly connected.
In view of this, we propose a modified version of singular perturbation:
\begin{equation}\label{aij-tilde}
\tilde{a}_{ij}=a_{ij}+\varepsilon I(n_{ij}>0),~ \varepsilon>0.
\end{equation}
If Condition B holds, then $\mathcal{G}(\tilde{A})$ must be strongly connected,
where $\tilde{A}=(\tilde{a}_{ij})$.
Henceforth, this model will be called the ``Bradley-Terry-$\varepsilon$" model.
The likelihood for the Bradley-Terry-$\varepsilon$ model is
\begin{equation}\label{Bradley-Terry-epsilon}
L_1(\mathbf{u})=\prod_{i,j=1}^n \left(  \frac{u_i}{u_i+u_j} \right)^{\tilde{a}_{ij}}=
\prod_{i,j=1}^n \left(  \frac{u_i}{u_i+u_j} \right)^{a_{ij}}\times \prod_{i,j=1}^n \left(  \frac{u_i}{u_i+u_j} \right)^{\varepsilon I(n_{ij}>0)}.
\end{equation}
As pointed by the referee, the item associated with $\varepsilon$ in the above expression can be viewed as the
penalized factor. Hereafter, the likelihood associated with the matrix $\tilde{A}$ will be called the penalized likelihood and
the solution $\mathbf{\hat{u}}=(\hat{u}_1, \ldots, \hat{u}_t)$ maximizing it, be labeled as the penalized MLE (PMLE).
Similar to Conner and Grant's (2000) proofs, we have:
\begin{theorem}\label{revision-bt}
For the penalized likelihood \eqref{Bradley-Terry-epsilon}: \\
(a)The PMLEs $\hat{u}_i$, $i=1,\ldots, t$ uniquely exist if and only if Condition B holds. \\
(b)if Condition $A$ holds, then the extended ranking is identical with the ranking provided by the original model.
\end{theorem}

\section{Generalized Bradley-Terry-$\varepsilon$ models}
When the objects do not differ in the quality or the sense of perception of a judge is not
sharp enough to detect the difference of objects,
ties (no preference) may be possible.
Based on that the Bradley-Terry model can be expressed as the logistic distribution:
$p_{ij}=\int_{-(\log u_i -\log u_j)}^\infty \mbox{sech}^2(x/2) dx $,
Rao and Kupper (1967) supposed that the judge expresses his
preference according to a random variable $d_{ij}$ with the distribution function
given by
\begin{equation*}
P(d_{ij}>d)=\frac{1}{4}\int_{-(\log u_i -\log u_j)+d}^\infty \mbox{sech}^2(x/2)dx.
\end{equation*}
When $|d_{ij}|$ is less than a threshold parameter, the judge is assumed to be unable to distinguish $i$ and $j$ and will declare
a tie. The Rao-Kupper model supposes that:
\begin{eqnarray*}
P(\mbox{$i$ beats $j$}) & = & u_i/(u_i+\theta u_j), \\
P(\mbox{$j$ beats $i$}) & = & u_j/(\theta u_i + u_j) \\
P(\mbox{$i$ ties $j$})  & = & (\theta^2-1) u_i u_j/[(u_i+\theta u_j)(u_j+\theta u_i)],
\end{eqnarray*}
where $\theta>1$ is a threshold parameter.
The existence of the MLE in the Rao-Kupper model aslo crucially depends on Condition A.
Using the single perturbation matrix $\tilde{A}$, the penalized likelihood for the Rao-Kupper-$\varepsilon$ model is
\begin{equation}\label{Rao-Kupper}
L_2(\mathbf{u},\theta)=\prod_{i<j} \left(\frac{ u_i }{ u_i +\theta u_j}\right)^{\tilde{a}_{ij}}
\left(\frac{u_j}{\theta u_i +u_j}\right)^{\tilde{a}_{ji}}
\left[\frac{ (\theta^2-1)u_i u_j }{(u_i+\theta u_j)(u_j+\theta u_i)} \right]^{t_{ij}},
\end{equation}
where $t_{ij}$ is the number of ties between $i$ and $j$.
We don't impose the singular perturbation on tie counts.
If there are no ties observed, it may be unnecessary to
use the tied models.

Extending for a different direction to allow tied observations,
Davidson (1970) proposed the geometric mean model,
in which the probabilities are in the ratio
\begin{equation*}
P(\mbox{$i$ beats $j$}): P(\mbox{$j$ beats $i$}): P(\mbox{$i$ ties $j$})=u_i:u_j:\theta\sqrt{u_iu_j}.
\end{equation*}
where $\theta\ge 0$ is the discrimination factor.
A characteristic of the Davidson model is that for
balanced paired comparisons, the ranking by the MLE agrees with that obtained from a scoring
system that allots two points for a win, one for a tie and zero for
a loss. The penalized likelihood for the Davidson-$\varepsilon$ model is
\begin{equation}\label{davidson}
L_3(\mathbf{u},\theta)=\prod_{i<j}\frac{ u_i^{\tilde{a}_{ij}} u_j^{\tilde{a}_{ji}} (\theta\sqrt{ u_i u_j })^{t_{ij}}}{ (u_i +u_j +\theta\sqrt{u_iu_j})^{\tilde{n}_{ij}}},
\end{equation}
where $\tilde{n}_{ij}=\tilde{a}_{ij}+\tilde{a}_{ji}+t_{ij}$.

Agresti (1990) postulated that the objects are ordered and the probability of $i$ beating
$j$ depends on which object is presented first.
If the objects are sports teams, this assumption leads to the ``home-field advantage" model:
\begin{equation*}
P(\mbox{$i$ beats $j$})=\left \{
\begin{matrix}
\gamma u_i/(\gamma u_i + u_j), & \mbox{if $i$ is at home} \\
u_i/( u_i + \gamma u_j),       & \mbox{if $j$ is at home},
\end{matrix}
\right.
\end{equation*}
where $\gamma>0$ measures the strength of the home-field advantage or disadvantage.
Let $n_{ij\cdot i}$ be the number of times that $i$ plays agaist $j$ with $i$ as the home team
and $a_{ij\cdot i}$ be the times of $i$ winning $j$ when $i$ is at home. The penalized likelihood for
the home-field-advantage-$\varepsilon$ model is
\begin{equation}\label{home-field-model}
L_4(\mathbf{u},\gamma)= \prod_{i<j}\left[ \frac{(\gamma u_i)^{\tilde{a}_{ij\cdot i}} u_j^{\tilde{a}_{ji\cdot i}}
}{(\gamma u_i + u_j)^{\tilde{a}_{ij\cdot i}+\tilde{a}_{ji\cdot i}}} \right]\left[
\frac{ u_i^{\tilde{a}_{ij\cdot j}} (\gamma u_j)^{\tilde{a}_{ji\cdot j}}
}{( u_i +\gamma u_j)^{\tilde{a}_{ij\cdot j}+\tilde{a}_{ji\cdot j}}} \right],
\end{equation}
where $\tilde{a}_{ij\cdot i}=a_{ij\cdot i} + \varepsilon I(n_{ij\cdot i}>0)$.

If ties and home-field advantage both exist, David (1988, page 144) suggested
\begin{eqnarray*}
P(\mbox{$i$ beats $j$})& = &\left \{
\begin{matrix}
\gamma u_i/(\theta u_i + u_j+\theta\sqrt{u_i u_j}), & \mbox{if $i$ is at home}, \\
u_i/( u_i + \theta u_j+\theta\sqrt{u_i u_j}),       & \mbox{if $j$ is at home}; \\
\end{matrix}
\right. \\
P(\mbox{$i$ ties $j$})& = &\left \{
\begin{matrix}
\theta\sqrt{u_i u_j}/(\gamma u_i + u_j+\theta\sqrt{u_i u_j}), & \mbox{if $i$ is at home}, \\
\theta\sqrt{u_i u_j}/( u_i + \gamma u_j+\theta\sqrt{u_i u_j}),       & \mbox{if $j$ is at home}.
\end{matrix}
\right.
\end{eqnarray*}
Let $t_{ij\cdot i}$ be the number of ties between $i$ and $j$ when $i$ is at home.
The penalized likelihood for the David-$\varepsilon$ model is
\begin{equation}\label{home-tie-model}
L_5(\mathbf{u},\theta,\gamma)
= \prod_{i<j}\left[ \frac{(\gamma u_i)^{\tilde{a}_{ij\cdot i}} u_j^{\tilde{a}_{ji\cdot i}} (\theta\sqrt{u_i u_j} )^{t_{ij\cdot i}}
}{(\gamma u_i + u_j+\theta\sqrt{u_i u_j})^{\tilde{n}_{ij\cdot i}}} \right]\left[
\frac{( u_i)^{\tilde{a}_{ij\cdot j}} (\gamma u_j)^{\tilde{a}_{ji\cdot j}} (\theta\sqrt{u_i u_j} )^{t_{ij\cdot i}}
}{( u_i +\gamma u_j+\theta\sqrt{u_i u_j})^{\tilde{n}_{ij\cdot j}}} \right],
\end{equation}
where $\tilde{n}_{ij\cdot i}=\tilde{a}_{ij\cdot i}+\tilde{a}_{ji\cdot i}+t_{ij\cdot i}$.

In order to guarantee the existence of the PMLE for the penalized likelihoods \eqref{home-field-model} and \eqref{home-tie-model},
we introduce the following condition that is a week version of Assumption 3 in Hunter (2004).
\vskip5pt
\noindent
Condition C. In every possible partition of the teams into two nonempty subsets $Q_1$ and $Q_2$,
some team in $Q_1$ has comparisons with some team in $Q_2$ as home team, and some team in $Q_1$ has
comparisons with some team in $Q_2$ as visiting team.

The existence and uniqueness of the PMLE in the above generalized Bradley-Terry-$\varepsilon$ models
is stated as follows:
\begin{theorem}\label{theorem-generalizedBT}
Let $\Omega=\{\mathbf{u}\in R^t:  u_i>0, \sum_{i=1}^t u_i=1\}$.
The parameter space is assumed to be
$\Omega\times\{\theta\in R: \theta>1\}$ for the penalized likelihood \eqref{Rao-Kupper};
$\Omega\times \{\theta\in R: \theta>0\}$ for the penalized likelihood \eqref{davidson};
$\Omega\times \{\gamma\in R: \gamma>0\}$ for the penalized likelihood \eqref{home-field-model};
and $\Omega\times\{\theta\in R: \theta>0\}\times \{\gamma\in R: \gamma>0\}$
for the penalized likelihood \eqref{home-tie-model}.\\
(a)The PMLE for the penalized likelihoods \eqref{Rao-Kupper} and \eqref{davidson} exists and is unique
if and only if Condition B holds and there is at least one tie.\\
(b)The PMLE for the penalized likelihood \eqref{home-field-model} exists and is unique if Condition C holds.\\
(c)The PMLE for the penalized likelihood \eqref{home-tie-model} exists and is unique if Condition C holds and there is at least one tie.
\end{theorem}

Theorem \ref{theorem-generalizedBT} immediately comes from the following two lemmas whose proofs are given
in Appendix A. The first shows the existence of the PMLE. The second proves the uniqueness of the PMLE by
showing the penalized log-likelihood is strictly concave on the parameter space.
\begin{lemma}\label{lemma-part1}
Let $\Omega=\{\mathbf{u}\in R^t:  u_i>0, \sum_{i=1}^t u_i=1\}$.
The parameter space is assumed to be
$\Omega\times\{\theta\in R: \theta>1\}$ for the penalized likelihood \eqref{Rao-Kupper};
$\Omega\times\{\theta\in R: \theta>0\}$ for the penalized likelihood \eqref{davidson};
$\Omega\times\{\gamma\in R: \gamma>0\}$ for the penalized likelihood \eqref{home-field-model}
and $\Omega\times\{\theta\in R: \theta>0\}\times\{\gamma\in R:
\gamma>0\}$ for the penalized likelihood \eqref{home-tie-model}.\\
(a) The PMLE for the penalized likelihoods \eqref{Rao-Kupper} and \eqref{davidson} exists
if and only if Condition B holds and there is at least one tie.\\
(b) The PMLE for the penalized likelihood \eqref{home-field-model} exists if and only if Condition C holds. \\
(c) The PMLE for the penalized likelihood \eqref{home-tie-model} exists if and only if Condition C holds
and there is at least one tie.
\end{lemma}

\begin{lemma}\label{lemma-part2}
For the reparameterization $(\mathbf{u}, \theta, \gamma)\to (\boldsymbol{\beta}, \log \theta, \log \gamma)$ in which $\beta_i=\log u_i-\log u_1$
and $\phi=\log \theta$. Let $\Omega^*=\{\beta\in R^t: \beta_1=0 \}$.\\
(a) The reparameterized versions of penalized log-likelihoods $\log L_2(\boldsymbol{\beta}, \log \theta)$ on $\Omega^*\times R_+$
and $\log L_3(\boldsymbol{\beta}, \log \theta)$ on $\Omega^*\times R$ are strictly concave
if and only if Condition B holds and there is at least one tie, where $R_+=\{x\in R:x>0\}$.\\
(b) The reparameterized version of the penalized log-likelihood
$\log L_4(\boldsymbol{\beta}, \log \gamma)$ is strictly concave on $\Omega^* \times R$
if Condition C holds.\\
(c) The reparameterized version of the penalized log-likelihood
$\log L_5(\boldsymbol{\beta}, \log \theta, \log \gamma)$ is strictly concave on $\Omega^* \times R^2$
if Condition C holds and there is at least one tie.
\end{lemma}

\begin{remark}
Lemma \ref{lemma-part1} states that Condition C is a necessary and sufficient condition
for the existence of PMLEs of penalized likelihoods \eqref{home-field-model}
and \eqref{home-tie-model}, while Lemma \ref{lemma-part2} only indicates that it is only a sufficient condition.
It is interesting to see if it is also necessary.
\end{remark}

\section{Numerical studies}
~~~~\textit{Example 1.} We first investigate the influence of a variety of $\varepsilon$ on the ranking
based on the PMLE in the Bradley-Terry-$\varepsilon$ model, fitted in the data example with $4$ objects given in Section 1.
By choosing $\varepsilon=0.001, 0.01, 0.1, 0.5, 1, 2$,
the PMLEs $\hat{u}_2, \hat{u}_3, \hat{u}_4$ ($\hat{u}_1=1$) are reported in Table \ref{table-exmple1}.
This table shows that although $\varepsilon$ differs, the ranking is the same ($1\succ2\succ 4\succ 3$).
Further, the ratio $\hat{u}_i/\hat{u}_j$ increases as $\varepsilon$ decreases, indicating that
$(\hat{u}_1, \hat{u}_2, \hat{u}_3, \hat{u}_4)$ approaches to the bound of the parameter space as $\varepsilon\to 0$.

\begin{table}[h!]\centering
\caption{Fitted parameters for different $\varepsilon$ ($\hat{u}_1=1$).}
\label{table-exmple1}
\begin{tabular}{lll lll lll l}
\hline
$\varepsilon$      & $0.001$    & $0.01$ & $0.1$ & $0.5$ & $1$  & $2$  \\
\hline
$\hat{u}_2$ &            $ 0.500$        &$ 0.502$&$ 0.524$&$ 0.600$&$ 0.667$ & $0.750$  &   \\
$\hat{u}_3$ &            $ 5.0\times 10^{-4}$&$ 0.005$     &$ 0.048$     &$ 0.200$     &$ 0.333$      & $0.500$ \\
$\hat{u}_4$ &            $ 0.001$             &$ 0.010$     &$ 0.091$     &$ 0.333$     &$ 0.500$      & $0.667$ \\
\hline
\end{tabular}
\end{table}

\textit{Example 2.} We construct a complicated win-loss matrix with $10$ objects $B_1, \ldots, B_{10}$
to see the performance of the improved singular perturbation.
The outcomes were assumed to be:

\begin{footnotesize}
\begin{equation*}
A=
\begin{pmatrix}
0& 2& 0& 0& 1& 1& 0& 1& 0& 0\\
1& 0& 2& 0& 0& 0& 1& 0& 1& 0\\
0& 1& 0& 1& 0& 0& 0& 1& 0& 1\\
0& 0& 0& 0& 2& 0& 0& 0& 1& 1\\
0& 0& 0& 1& 0& 1& 0& 0& 0& 1\\
0& 0& 0& 0& 0& 0& 2& 0& 0& 1\\
0& 0& 0& 0& 0& 1& 0& 2& 0& 0\\
0& 0& 0& 0& 0& 0& 1& 0& 1& 0\\
0& 0& 0& 0& 0& 0& 0& 0& 0& 2\\
0& 0& 0& 0& 0& 0& 0& 0& 1& 0\\
\end{pmatrix}
\end{equation*}
\end{footnotesize}

\hspace{-6.5mm}Let $E_1=\{B_1, \ldots, B_5\}$ and $E_2=\{B_6, \ldots, B_{10}\}$.
Since the objects in $E_1$ beats the objects in $E_2$,
the ranking of objects in $E_1$ should be higher than the ranking of those in $E_2$.
Since the respective results for ``$B_1$ vs. $B_2$", ``$B_2$ vs. $B_3$" ``$B_3$ vs. $B_4$" and ``$B_4$ vs. $B_5$" are $2:1$,
$2:1$, $1:0$ and $2:1$, and there is only another game between $B_1$ and $B_5$ with the winner $B_1$,
the ranking of $B_1$, $\ldots$, $B_5$ in $E_1$ should be $B_1\succ B_2\succ B_3\succ B_4\succ B_5$.
With the same argument, the ranking of in $E_2$ should be $B_6\succ B_7\succ B_8\succ B_9\succ B_{10}$.
It is interesting to see if the improved singular perturbation method gives the ranking consistent with the subjective ranking.
We chose $\varepsilon=t^{-1/2}, (\log t/t)^{1/2}, 0.8, 1, 2$ ($t=10$) and the PMLEs are given in Table \ref{table-example2}.
This table shows that the ranking based on the PMLE coincides with the objectively ranking perfectly.

\begin{table}[h!]
\centering
\caption{Fitted parameters for different $\varepsilon$ ($\hat{u}_{10}=1$).}
\label{table-example2}
\begin{tabular}{lll lll lll}
\hline
$\varepsilon$      & $t^{-1/2}$  & $(\log t/t)^{1/2}$ & $0.8$ & $1$ & $2$  \\
\hline
$\hat{u}_1$ &            $ 17.13$&$ 9.414$&$ 5.062$&$ 4.017$&$ 2.277$     \\
$\hat{u}_2$ &            $ 11.23$&$ 6.558$&$ 3.815$&$ 3.131$&$ 1.945$     \\
$\hat{u}_3$ &            $ 8.353$&$ 5.133$&$ 3.166$&$ 2.660$&$ 1.758$      \\
$\hat{u}_4$ &            $ 5.348$&$ 3.731$&$ 2.581$&$ 2.252$&$ 1.614$      \\
$\hat{u}_5$ &            $ 4.544$&$ 3.267$&$ 2.337$&$ 2.066$&$ 1.531$      \\
$\hat{u}_6$ &            $ 3.552$&$ 2.743$&$ 2.091$&$ 1.887$&$ 1.462$      \\
$\hat{u}_7$ &            $ 2.918$&$ 2.317$&$ 1.829$&$ 1.675$&$ 1.354$      \\
$\hat{u}_8$ &            $ 2.451$&$ 2.022$&$ 1.660$&$ 1.543$&$ 1.292$       \\
$\hat{u}_9$ &            $ 1.427$&$ 1.353$&$ 1.267$&$ 1.232$&$ 1.142$      \\
\hline
\end{tabular}
\end{table}

\textit{Example 3.} We consider the win-loss outcomes given by Mease (2003) with $5$ teams $B_1, \ldots, B_5$ and
the game results: $a_{13}=1$, $a_{15}=1$, $a_{21}=1$, $a_{25}=1$, $a_{34}=1$, $a_{35}=1$,
$a_{45}=2$ and for other $a_{ij}=0$. The PMLEs fitted in the Bradley-Terry-$\varepsilon$ model are given
in Table \ref{table-Mease2003}. The ranking of $B_1, \ldots, B_5$ according their PMLEs is
$B_2\succ B_1\succ B_3\succ  B_4\succ  B_5$, agreeing with the subjectively ranking by Mease (2003).

\begin{table}[h!]
\centering
\caption{Fitted parameters for different $\varepsilon$ ($\hat{u}_1=1$).}
\label{table-Mease2003}
\begin{tabular}{lll lll lll}
\hline
$\varepsilon$      & $0.01$  & $0.05$ & $0.1$ & $1$ & $2$  \\
\hline
$\hat{u}_2$ &            $ 50.00$&$ 10.03$&$ 5.122$&$ 1.339$&$ 1.176$     \\
$\hat{u}_3$ &            $ 0.030$&$ 0.154$&$ 0.298$&$ 0.867$&$ 0.931$     \\
$\hat{u}_4$ &            $ 0.001$&$ 0.030$&$ 0.104$&$ 0.772$&$ 0.886$      \\
$\hat{u}_5$ &            $ 0.000$&$ 0.003$&$ 0.017$&$ 0.421$&$ 0.607$      \\
\hline
\end{tabular}
\end{table}

\textit{2008 NFL regular season.}
The NFL that is the largest United States Football League
and has tens of millions of fans in each season,
is divided into two conferences:
American Football Conference (AFC) and National Football Conference ( NFC).
Each conference has $16$ teams,
and these $16$ teams are evenly divided into $4$ competition areas: East, South, West and North.
The data for 2008 NFL regular season that can be downloaded from
\url{http://en.wikipedia.org/wiki/2008_NFL_season}, contains the tied observations and
Condition A fails in this data since there is the team ``Detroit Lions" with no wins.
We chose $\varepsilon=(\log 32/32)^{1/2}=0.329^*$.
\footnote{$^*$When $\varepsilon=(\log t/t)^{1/2}$, the PMLE in the Bradley-Terry-$\varepsilon$ model
is uniformly consistent under some conditions, whose proof is given in Appendix B.}
\hspace{-5pt}The fitted parameters using the David-$\varepsilon$ model,
Rao-Kupper-$\varepsilon$ model and Davidsion-$\varepsilon$ model are presented in Table \ref{NFL-2008-merits}.

\begin{table}[h!]\centering
\caption{Merits of NFL 2008, the values in parentheses and brackets are fitted parameters by the Rao-Kupper-$\varepsilon$ model
and Davidson-$\varepsilon$ model, respectively. The values without parentheses and brackets
are fitted parameters by the David-$\varepsilon$ model.
} \label{NFL-2008-merits} \vskip5pt
\scriptsize
\begin{tabular}{llc | llc }
\hline
division &  team                & merit   & division  &  team                 &  merit \\
AFC East &  New England Patriots & 1.054(1.030)[1.030]  & AFC South &  Indianapolis Colts   & 1.579(1.350)[1.346] \\
         &  New York Jets        & 0.737(0.506)[0.504]  &           &  Houston Texans       & 0.844(0.477)[0.475] \\
         &  Miami Dolphins       & 1.000(1.000)[1.000]  &           &  Tennessee Titans     & 1.713(1.727)[1.723] \\
         &  Buffalo Bills        & 0.517(0.258)[0.258]  &           &  Jacksonville Jaguars & 0.529(0.209)[0.208]  \\
AFC North&  Cincinnati Bengals   & 0.518(0.360)[0.349]  & AFC West  &  San Diego Chargers   & 0.723(0.312)[0.312] \\
         &  Baltimore Ravens     & 1.430(1.769)[1.768]  &           &  Denver Broncos       & 0.633(0.236)[0.236]  \\
         &  Pittsburgh Steelers  & 1.694(2.550)[2.550]  &           &  Oakland Raiders      & 0.450(0.133)[0.133] \\
         &  Cleveland Browns     & 0.503(0.315)[0.314]  &           &  Kansas City Chiefs   & 0.274(0.055)[0.055] \\
NFC East &  Dallas Cowboys       & 0.896(0.436)[0.437]  & NFC North &  Minnesota Vikings    & 1.053(0.496)[0.496]  \\
         &  Philadelphia Eagles  & 1.018(0.578)[0.596]  &           &  Green Bay Packers    & 0.545(0.151)[0.151]  \\
         &  New York Giants      & 1.498(1.176)[1.181]  &           &  Chicago Bears        & 0.838(0.318)[0.319]  \\
         &  Washington Redskins  & 0.731(0.293)[0.294]  &           &  Detroit Lions        & 0.199(0.022)[0.022]     \\
NFC South&  New Orleans Saints   & 0.694(0.129)[0.129]  & NFC West  &  Arizona Cardinals    & 0.781(0.248)[0.249] \\
         &  Atlanta Falcons      & 1.075(0.264)[0.265]  &           &  San Francisco 49ers  & 0.514(0.105)[0.105] \\
         &  Carolina Panthers    & 1.348(0.404)[0.405]  &           &  Seattle Seahawks     & 0.350(0.050)[0.050]  \\
         &  Tampa Bay Buccaneers & 0.794(0.165)[0.165]  &           &  St. Louis Rams       & 0.265(0.023)[0.023]      \\
\hline\\
\multicolumn{6}{l}{For the David-$\varepsilon$ model, $\widehat{\theta}=0.006$, $\widehat{\gamma}=1.221$}\\
\multicolumn{6}{l}{For the Rao-Kupper-$\varepsilon$ model, $\widehat{\theta}=1.001$} \\
\multicolumn{6}{l}{For the Davidson-$\varepsilon$ model, $\widehat{\theta}=0.002$}
\end{tabular}
\end{table}

The PMLEs fitted by the Rao-Kupper-$\varepsilon$ model and Davidsion-$\varepsilon$ model
are very close with the largest absolute difference $0.11$. It may be not surprised since
there is only one tie, leading a very small tie parameter $\hat{\theta}=1.001$ and $\hat{\theta}=0.002$
for the Rao-Kupper-$\varepsilon$ model and Davidsion-$\varepsilon$ model, respectively.
When considering the home-field advantage, there are some major difference about the MLE
fitted by the David-$\varepsilon$ model and other two models.
This indicates that there may exist non-ignorable home-field advantage effect.

It is interesting to compare if the chosen six playoff seeds of AFC and NFC by the win-loss percentage records (PCTs)
and by the merits obtained in the David-$\varepsilon$ model are the same according to the NFL rule.
The former is employed by the seeding system of the NFL.
The NFL rule can be briefly summarized as follows:
the four division champions from each conference
(the team in each division with the best PCT),
which are seeded one through four based on their PCTs;
two wild card qualifiers from each conference
(those non-division champions with the conference's best record),
which are seeded five and six. From Table \ref{table-seeds-NFL2008},
only the third seed in AFC is different, with Miami Dolphins by PCT and New England Patriots by merits.

\vskip12pt
\begin{table}[h!]\centering
\caption{The six playoff seeds of AFC and NFC. The values in brackets are PCTs.}
\label{table-seeds-NFL2008}
\vskip5pt
\scriptsize
\begin{tabular}{cll}
Seed & Based on PCTs                    & Based on merits \\
\hline
\multicolumn{3}{c}{AFC seeds}\\
1    & Tennessee Titans(South) $[.813]$    &Tennessee Titans(South) $[1.713]$     \\
2    & Pittsburgh Steelers (North) $[.750]$ &Pittsburgh Steelers(North) $[1.694]$      \\
3    & Miami Dolphins (East)  $[.688]$     &New England Patriots(East) $[1.054]$\\
4    & San Diego Chargers (West) $[.500]$   &San Diego Chargers(West)   $[0.723]$\\
5    & Indianapolis Colts    $[.750]$      &Indianapolis Colts     $[1.579]$ \\
6    & Baltimore Ravens      $[.688]$      &Baltimore Ravens       $[1.430]$\\
\hline
\multicolumn{3}{c}{NFC Seeds}\\
1    &New York Giants (East) $[.750]$      &New York Giants (East)  $[1.498]$\\
2    &Carolina Panthers (South) $[.750]$   &Carolina Panthers (South) $[1.348]$\\
3    &Minnesota Vikings (North) $[.625]$   &Minnesota Vikings (North) $[1.053]$\\
4    &Arizona Cardinals (West) $[.563]$    &Arizona Cardinals (West) $[0.781]$\\
5    &Atlanta Falcons         $[.688]$     &Atlanta Falcons $[1.075]$\\
6    &Philadelphia Eagles     $[.594]$     &Philadelphia Eagles $[1.018]$\\
\hline
\end{tabular}
\end{table}

\section{Bayesian ranking}
Caron and Doucet (2012) introduced latent variables and Gamma prior distributions to make
Bayesian inference for the Bradley-Terry model and its extensions. Some brief descriptions are given in the following.
They assumed that $\mathbf{u}$ has the prior density, i.e., a product of $t$ independent Gamma densities with
a shape parameter $d$ and a rate parameter $b$:
\begin{equation*}
p(\mathbf{u})=\prod_{i=1}^t \frac{b^d}{\Gamma(d)} u_i^{d-1} e^{-bu_i}.
\end{equation*}
where $\Gamma(\cdot)$ is the gamma function.
The resulting posterior distribution of $\mathbf{u}$ is:
\begin{equation*}
P(\mathbf{u}|A ) \propto P(A|\mathbf{u})P(\mathbf{u})= \prod_{i=1}^t\binom{n_{ij}}{a_{ij}} (\frac{ u_i}{ u_i + u_j })^{a_{ij}} (\frac{ u_j}{ u_i + u_j })^{a_{ji}}
\times \prod_{i=1}^t\frac{b^d}{\Gamma(d)} u_i^{d-1} e^{-bu_i}.
\end{equation*}
By maximizing $P(\mathbf{u}|A )$, it yields the MAP estimate of $\mathbf{u}$
that can be solved by the EM algorithm. From the above expression, it is easy to show the existence of MAP estimate of $\mathbf{u}$
when $d>1$ and $b>0$.
For each pair $(i,j)$
and each associated paired comparison
$k=1, \ldots, n_{ij}$, let $Y_{ki}\sim \mathcal{E}( u_i )$ and $Y_{kj}\sim \mathcal{E}(u_j)$
where $\mathcal{E}(u)$ is the exponential distribution of rate parameter $u$. The latent random variables are given by
\begin{equation*}
Z_{ij} = \sum_{k=1}^{n_{ij} } \min( Y_{kj}, Y_{ki} )\sim Gamma(n_{ij}, u_i+u_j),~~1\le i<j \le t,
\end{equation*}
where $Gamma(n_{ij}, u_i+u_j)$ is the Gamma distribution with
the shape parameter $n_{ij}$ and the rate parameter $u_i+u_j$,
such that
\begin{equation*}
p(\mathbf{z}|A, \mathbf{u})= \prod_{\{(i,j):1\le i<j \le t, n_{ij}>0\}} \frac{ (u_i+u_j)^{n_{ij}} }{\Gamma(n_{ij})}
z_{ij}^{n_{ij}-1} e^{-(u_i+u_j)z_{ij} }.
\end{equation*}
The resulting complete log-likelihood is
\begin{equation*}
 \sum_{i=1}^t a_i \log u_i -
\sum_{\{(i,j):1\le i<j \le t, n_{ij}>0\} }[ (u_i+u_j)z_{ij}-(n_{ij}-1)\log z_{ij} + \log \Gamma (n_{ij})],
\end{equation*}
where $a_i=\sum_{j=1}^t a_{ij}$.
The EM algorithm proceeds as follows at iteration $t$:
\begin{equation*}
\mathbf{u}^{(t)}=\arg \max_{\mathbf{u}}\left\{ \sum_{i=1}^t[(a-1+a_i)\log u_i - bu_i]-
\sum_{1\le i< j\le t}(u_i + u_j)\frac{n_{ij}}{u_i^{(t-1)}+u_j^{(t-1)}} \right\}.
\end{equation*}
It follows that
\begin{equation*}
\lambda_i^{(t)}=\frac{ d-1 + a_i }{ b + \sum_{j\neq i} \frac{n_{ij}}{u_i^{(t-1)}+u_j^{(t-1)} } }.
\end{equation*}
When $d=1$ and $b=0$, then the MAP and ML estimates coincide.
The similar latent variables can be incorporated into
the Rao-Kupper model and home-field-advantage model. For details, see Section 3 in Caron and Doucet (2012).
For the sake of simplicity, we only made the comparisons between the Bayesian ranking and perturbation ranking in the Bradley-Terry model
by using Examples 1-3.

We used the computational programs of Caron and Doucet (2012) that are available at
\url{http://amstat.tandfonline.com.proxygw.wrlc.org/doi/suppl/10.1080/10618600.2012.638220\#tabModule},
in which $b=dt-1$ by default. We chose $d=1.1, 1.5, 2, 3$ and an automatically estimated value by their programs.
The fitted MAP estimates in Example 1 are given in Table \ref{table-exmple1-Bayesian}.
In this table, when $d=1.1, 1.5, 0.961$, the rankings are $1\succ 2 \succ 4 \succ 3$
based on MAP estimates, according with those based on the improved perturbation method.
But when $d=2,3$, it leads to a different ranking ($1\succ 4 \succ 2 \succ 3$).
For Examples 2 and 3, the rankings based on MAP estimates are identical to those of the improved perturbation method for all chosen $d$
(the results are not shown here).

\begin{table}[h!]\centering
\caption{Fitted parameters for different $d$.}
\label{table-exmple1-Bayesian}
\begin{tabular}{lll lll }
\hline
$d$         &            $1.1$      & $1.5$       & $2$         & $3$         &$0.961^*$\\
\hline
$\hat{u}_1$ &            $ 0.572$   &$ 0.430$     &$ 0.374$     &$ 0.330$     &$ 0.391$     \\
$\hat{u}_2$ &            $ 0.281$   &$ 0.229$     &$ 0.221$     &$ 0.224$     &$ 0.228$     \\
$\hat{u}_3$ &            $ 0.053$   &$ 0.139$     &$ 0.175$     &$ 0.204$     &$ 0.167$     \\
$\hat{u}_4$ &            $ 0.094$   &$ 0.201$     &$ 0.230$     &$ 0.243$     &$ 0.215$     \\
\hline
\multicolumn{6}{c}{\scriptsize$^*$estimated by Caron and Doucet's programs.}
\end{tabular}
\end{table}

\section{Discussion}

If paired comparisons are balanced, the ranking based on the MLE in the Bradley-Terry-$\varepsilon$ model
is identical to that based on the scores $a_i=\sum_{j=1}^t a_{ij}$. This can be easily shown by noting that
if all $n_{ij}$ are equal to $n$,
\begin{equation}
a_i-a_j=n\sum_{k=1}^t\frac{\hat{u}_k(\hat{u}_i-\hat{u}_j)}{(\hat{u}_i+\hat{u}_k)(\hat{u}_j+\hat{u}_k)}.
\end{equation}
This implies that for balanced paired comparisons, the singular perturbation
is robust to $\varepsilon$, i.e., the ranking is regardless of the choice of $\varepsilon>0$.
For nonbalanced paired comparisons,
Examples 1-3 show that if Condition B holds,
there is also the same conclusion.
However, it may be challenging to show this conjecture.
Smaller is $\varepsilon$, the ranking decided by the MLE is more obvious since the MLE
approaches to the bound of the parameter space more closely as $\varepsilon$ decreases.
For real applications, we could choose $\varepsilon=(\log t/t)^{1/2}$ or some others values less than $1$.
Moreover, similar phenomena that the ranking from the generalized Bradley-Terry-$\varepsilon$ models
is identical to that from original models when the MLEs exist and is robust to $\varepsilon$ when Condition B or C holds,
were observed in simulation studies that are not shown here. It is interesting to show this conclusion.

Davidson and Solomon (1973) proposed a Bayesian approach to analyze the
general binomial model for paired comparisons by using the Beta prior distribution
and, in particular, applied it to the Bradley-Terry model.
Let $\mathbf{p}=\{p_{ij}\}_{i,j=1,\ldots,t}$,
where $p_{ij}$ is the probability of preference for object $i$ over object $j$,
$p_{ij}+p_{ji}=1$, and $p_{ii}=1/2$ for convenience.
Davidson and Solomon used conjugate Beta prior distributions (a straightforward extension of the single binomial model):
\begin{equation*}
p_{ij}\sim Beta( a_{ij}^0, a_{ji}^0 ),~~i,j=1,\ldots,t,
\end{equation*}
where $A^0=(a_{ij}^0)_{i,j=1,\ldots,t}$ is a matrix of prior parameters satisfying $a_{ij}^0>-1$ and $a_{ii}^0=0$.
By applying it to the Bradley-Terry model, it is obtained:
\begin{equation*}
P(\mathbf{u}|A )=\prod_{i<j} (\frac{u_i}{u_i+u_j})^{a_{ij}+a_{ij}^0} (\frac{u_j}{u_i+u_j})^{ a_{ji}+a_{ji}^0 },
\end{equation*}
omitting the constant item.
In order to guarantee the existence of the MAPE, they supposed that $\mathcal{G}(A^0)$ is strongly connected.
However, the properties of the extended MLE had not explored in their paper.
Annis and Wu (2006) chose $a_{ij}^0=1/[2(t-1)]$ as the prior parameter.
By taking $a_{ij}^0=\varepsilon$, it leads to Conner and Grant's perturbation method.
Similarly, the improved perturbation version is derived by letting $a_{ij}^0=\varepsilon I(n_{ij}>0)$.
On the other hand, if we disturb $A$ by adding a general matrix $A^0$,
it produces Davidson and Solomon's Bayesian approach.
By applying different Gamma prior distributions to the merits parameters,
Caron and Doucet (2012) obtained MAP estimates of merits parameters in the Bradley-Terry model
and its extensions such as the home-field advantage model, Rao-Kupper model and some other related models.
Based on the MAP estimate, one can derive a ranking of objects.
Examples 2 and 3 display the same rankings between the PMLEs and the Caron and Doucet's MAP estimates
while Example 1 reveals different rankings for some MAP estimates.
It is interesting to explore the robustness of their method and see if the ranking based on the MAP estimate
inherits from that of the original model.

\section*{Acknowledgement}
The author would like to thank the Editor and the referee for their valuable
comments. This work was partly done while the author was in the George Washington
University and is partly supported by National Science Foundation of China (No. 11341001).

\section*{Appendix A}

Parts of proofs can be found in Ford (1957) and Hunter (2004).

\begin{proof}[Proof of Lemma \ref{lemma-part1}]
The proofs of Lemma \ref{lemma-part1} (a) and (b) are similar to that of (c).
Thus we only give the proof of Lemma \ref{lemma-part1} (c).

(c)We first prove that if Condition C holds and there is at least one tie, then
the MLE exists.
Suppose that $(\bar{u}, \bar{\theta}, \bar{\gamma})$ is on the boundary of the parameter space
$\Omega\times (0, \infty)\times (0, \infty)$.
It suffices to show
\begin{equation}\label{part1-existence}
\lim_{(u,\theta, \gamma)\to(\bar{u},\bar{\theta}, \bar{\gamma})}L_5(\mathbf{u}, \theta, \gamma) \to 0.
\end{equation}
If $\mathbf{u}$ lies on the boundary of $\Omega$, then $u_i=0$ and $u_j>0$ for some $i$ and $j$.
Since Condition C implies Condition B, there exists a directed path from $i$ to $j$ in the
directed graph constructed by the adjacent matrix
$(\tilde{a}_{ij\cdot i}+ \tilde{a}_{ij\cdot j})_{ij}$.
Therefore, there must be some object $k$ with $u_k=0$ that points to some individual $l$
in the direct path from $i$ to $j$ with
$u_l>0$ and $\tilde{a}_{kl\cdot k}+ \tilde{a}_{kl\cdot l} >\varepsilon $.
Therefore, as $u_k\to 0$, we have
\begin{equation*}
\lim_{(u,\theta, \gamma)\to(\bar{u},\bar{\theta}, \bar{\gamma})}L_5(u,\theta, \gamma)
\le \lim_{(u,\theta, \gamma)\to(\bar{u},\bar{\theta}, \bar{\gamma})}[\frac{ u_k }{(\theta u_k + u_l + \gamma\sqrt{u_k u_l})}]^\varepsilon \psi(\mathbf{u},\theta, \gamma) \to 0.
\end{equation*}

Since Condition C implies that there exists $\tilde{a}_{ij\cdot i}\ge \varepsilon$, we have
\begin{equation*}
L_5(\mathbf{u}, \theta, \gamma)
\le [\frac{ \theta u_i }{(\theta u_i + u_j + \gamma\sqrt{u_i u_j})}]^\varepsilon \psi(\mathbf{u},\theta, \gamma).
\end{equation*}
The above expression does go to zero as $\theta \to 0$.
On the other hand, since Condition C implies that there exists $\tilde{a}_{ji\cdot i}\ge \varepsilon$,
we have
\begin{equation*}
L_5(\mathbf{u},\theta, \gamma)
\le [\frac{ u_i }{(\theta u_i + u_j + \gamma\sqrt{u_i u_j})}]^\varepsilon \psi(u,\theta, \gamma).
\end{equation*}
The above expression does go to zero as $\theta \to \infty$.

Since there at least one tie (assume that it is $t_{ij\cdot i}\ge 1$), if $\gamma\to 0$, we have
\begin{equation*}
\lim_{(\mathbf{u},\theta, \gamma)\to(\bar{u},\bar{\theta}, \bar{\gamma})}L_5(\mathbf{u},\theta, \gamma)
\le \lim_{(\mathbf{u},\theta, \gamma)\to(\bar{u},\bar{\theta}, \bar{\gamma})}[\frac{ \gamma\sqrt{u_i u_j} }{(\theta u_i + u_j + \gamma\sqrt{u_i u_j})}] \psi(\mathbf{u},\theta, \gamma) \to 0;
\end{equation*}
if $\gamma \to \infty$, we have
\begin{equation*}
\lim_{(\mathbf{u},\theta, \gamma)\to(\bar{u},\bar{\theta}, \bar{\gamma})}L_5(\mathbf{u},\theta, \gamma)
\le \lim_{(\mathbf{u},\theta, \gamma)\to(\bar{u},\bar{\theta}, \bar{\gamma})}[\frac{ \theta u_i }{(\theta u_i + u_j + \gamma\sqrt{u_i u_j})}]^\varepsilon \psi(u,\theta, \gamma) \to 0.
\end{equation*}

Next, we show the necessary of Condition C and that there is at least one tie.
If there is no tie, then $\lim_{\gamma\to 0} L_5 (\mathbf{u}, \theta, \gamma)\to 0$ such that the MLE doesn't exist.
There are three cases to consider Condition C:\\
(1)If there there is a partition of all objects to two sets $E_1$ and $E_2$ such that
there are no interset comparisons, then we doublet $u_i, i\in E_2$ (assume that the referred object is in $E_1$),
the likelihood $L_5(\mathbf{u},\theta, \gamma)$ doesn't change such that the MLE doesn't exist.\\
(2)If there is a partition of all objects to two sets $E_1$ and $E_2$
such that all objects in $E_1$ are at home and all interset comparisons are won by objects from $E_1$,
then the likelihood $L_5(\mathbf{u},\theta, \gamma)$ must attain its maximum at $\theta=\infty$.\\
(3)If there is a partition of all objects to two sets $E_1$ and $E_2$
such that all objects in $E_1$ are at home and all interset comparisons are won by objects from $E_2$
then the likelihood $L_5(\mathbf{u},\theta, \gamma)$ must attain its maximum at $\theta=0$.\\
This completes the proof.
\end{proof}

\begin{proof}[Proof of Lemma \ref{lemma-part2}]
After some modifications, the arguments of the proof of Hunter (2004) can be easily
extended to the cases of (a). The proof of (b) is similar to that of
(c) and we only give the latter.

Let $\ell (\boldsymbol{\beta}, \log \theta, \log \gamma)$ be the log-likelihood $L_5(\mathbf{u}, \theta,\gamma)$:
\begin{eqnarray*}
\begin{array}{lll}
\ell(\boldsymbol{\beta}, \log \theta, \log \gamma)=
\sum_{i<j}[  a_{ij\cdot i} (\log \theta + \beta_i) + a_{ji\cdot i} \beta_j  + t_{ij\cdot i}[\log \gamma
+\beta_i/2 +\beta_j/2] \\
~~~~~~~~~~~~~~~~~~~~~~-n_{ij\cdot i}\log ( e^{\log \theta +\beta_i} + e^{\beta_j} +e^{\log \gamma + \beta_i/2+\beta_j/2} )
+a_{ij\cdot j}\beta_i + a_{ji\cdot j}(\log \theta + \beta_j)
\\
~~~~~~~~~~~~~~~~~~~~~~+t_{ij\cdot j}(\log \gamma  + \beta_i/2 +\beta_j/2)
-n_{ij\cdot j}(e^{\beta_i} + e^{\log \theta + \beta_j} +e^{\log \gamma + \beta_i/2 + \beta_j/2})
 ].
\end{array}
\end{eqnarray*}
By the inequality (21) in Hunter (2004), we have that
for positive numbers $c_1, \ldots, c_N$ and $d_1, \ldots, d_N$ and $p\in (0,1)$,
\begin{equation}\label{inequality-Holder}
\log \sum_{k=1}^N c_k^pd_k^{1-p} \le p \log \sum_{k=1}^N c_k + (1-p)\log \sum_{k=1}^N d_k,
\end{equation}
with equality if and only if there exists some $\xi>0$ such that
$c_k=\xi d_k$ for all $k$.
By definition, a log-likelihood $\lambda$ is concave if, for any parameter vectors $\boldsymbol{\alpha}, \boldsymbol{\beta}$ and $p\in (0, 1)$,
\begin{equation}\label{likelihood-concavity}
\ell [p\boldsymbol{\alpha} + (1-p)\boldsymbol{\beta} ]
\ge p\lambda(\boldsymbol{\alpha}) + (1-p)\lambda(\boldsymbol{\beta});
\end{equation}
concavity is strict if $\alpha\neq \beta$ implies that the inequality \eqref{likelihood-concavity} is strict.
Inequality \eqref{inequality-Holder} implies that
\begin{equation}
\begin{array}{l}
~~~~-\log \{e^{p\log \theta_1 +(1-p)\log \theta_2 + p\alpha_i + (1-p)\beta_i}
+e^{p\alpha_j+(1-p)\beta_j} + \\
~~~~~~~~~~~~~~~~~e^{p\log \gamma_1 + (1-p)\log \gamma_2 + \frac{1}{2}[
p(\alpha_i+\alpha_j)+(1-p)(\beta_i+\beta_j)] } \} \\
\ge -p\log [e^{\log \theta_1 + \alpha_i} + e^{\log \alpha_j} + e^{\log \gamma_1
+ \frac{1}{2}(\alpha_i+\alpha_j)}] \\
~~~~ - (1-p)\log [e^{\log \theta_2 +\beta_i}+e^{\log \beta_j} + e^{
\log \gamma_2 + \frac{1}{2}(\beta_i+\beta_j)} ].
\end{array}
\end{equation}
Hence multiplying the above inequality by $\tilde{a}_{ij\cdot i}+t_{ij\cdot i}/2$ and then summing over $i$ and $j$
demonstrates the concavity of the log-likelihood $\ell (\boldsymbol{\beta}, \log \theta, \log \gamma)$.
By the equality condition for \eqref{inequality-Holder},
whenever $\max\{\tilde{a}_{ij\cdot i}, t_{ij\cdot i}\}>0$, we have
$$\log \theta_1 - \log \theta_2 + \alpha_i-\beta_i = \alpha_j-\beta_j
=\log \gamma_1 - \log \gamma_2 + (\alpha_i-\beta_i)/2 + (\alpha_j-\beta_j)/2;$$
whenever $\max\{\tilde{a}_{ij\cdot j}, t_{ij\cdot j}\}>0$, we have
$$\log \theta_1 - \log \theta_2 + \alpha_j-\beta_j = \alpha_i-\beta_i
=\log \gamma_1 - \log \gamma_2 + (\alpha_i-\beta_i)/2 + (\alpha_j-\beta_j)/2.$$
If there is no team that has comparisons both as home team and as visiting team,
then we can partition all teams into two sets with the first set including all home teams
and the second including all visiting teams such that Condition C fails.
Therefore Condition C implies that there exists a loop such that
$\log\theta_1-\log \theta_2+\beta_{i_1}-\alpha_{i_1}=\beta_{i_1}-\alpha_{i_1}$ for some $i_1$,
which means that $\theta_1=\theta_2$.
Moreover, Condition C implies $\boldsymbol{\beta}=\boldsymbol{\alpha}$. Consequently, $\gamma_1=\gamma_2$.
This completes the proof.
\end{proof}

\section*{Appendix B}
\begin{theorem}\label{theorem-consistency}
Let $M_t=\max_{i,j} u_i/u_j$, $n_i=\sum_j n_{ij}$ and $c_{ij}=\#\{k: n_{ik}>0, n_{kj}>0\}$.
Assume that $\max_{i,j}n_{ij}\le N$ where $N$ is a positive constant.
If $M_t=o(\sqrt{t/\log t})$, $\min_{i,j}c_{ij}/t\ge \tau>0$ and $\varepsilon=\sqrt{\log t/t }$,
then the MLE is uniformly consistent in the sense that $\max_{i}|\hat{u}_i/u_i -1|=o_p(1)$.
\end{theorem}

\begin{proof}[Proof of Theorem \ref{theorem-consistency}]
Let $a_i=\sum_j a_{ij}$ and $\tilde{a}_i=\sum_j \tilde{a}_{ij}$.
This proof is similar to that of Theorem 1 in Yan et, al. (2012) by noting that
\begin{eqnarray*}
\max_{i=1,\ldots,t}|\tilde{a}_i-E(a_i)|&  \le & \max_{i=1,\ldots,t}(|\tilde{a}_i-a_i|+|a_i-E(a_i)|)
 \le  \max_{i=1,\ldots,t}|a_i-E(a_i)| +\varepsilon t \\
 &\le& \max_{i=1,\ldots,t}|a_i-E(a_i)|+\sqrt{t\log t},
\end{eqnarray*}
and we omit the details.
\end{proof}

\end{document}